\documentclass{amsart}
\usepackage{amsmath,amsthm,amssymb,amscd}
\usepackage{graphicx}
\usepackage{tikz-cd}
\usepackage{paralist}
\usepackage{environ}
\usepackage{xcolor}
\usepackage{hyperref}
\usepackage{centernot}
\usepackage{mathtools}

\usepackage[framemethod=TikZ]{mdframed}

\usepackage{marvosym}

\usepackage{xcolor,cancel}

\usepackage{booktabs,multirow}

\usepackage{subfigure}
\renewcommand{\thesubfigure}{\thefigure.\arabic{subfigure}}
\makeatletter
\renewcommand{\p@subfigure}{}
\renewcommand{\@thesubfigure}{\thesubfigure:\hskip\subfiglabelskip}
\makeatother

\makeatletter
\def\square{\pst@object{square}}% reads star and options and continues with \square@i
\def\square@i(#1,#2)#3{{\use@par\solid@star\psframe[origin={#1,#2}](#3,#3)}}
\makeatother

\makeatletter
\DeclareFontFamily{U}{tipa}{}
\DeclareFontShape{U}{tipa}{bx}{n}{<->tipabx10}{}
\newcommand{\arc@char}{{\usefont{U}{tipa}{bx}{n}\symbol{62}}}%

\newcommand{\arc}[1]{\mathpalette\arc@arc{#1}}

\newcommand{\arc@arc}[2]{%
  \sbox0{$\m@th#1#2$}%
  \vbox{
    \hbox{\resizebox{\wd0}{\height}{\arc@char}}
    \nointerlineskip
    \box0
  }%
}
\makeatother

\makeatletter
\newcommand{\doublewedge}{\big@doubleop{\wedge}}
\newcommand{\big@doubleop}[1]{%
  \DOTSB\mathop{\mathpalette\big@doubleop@aux{#1}}\slimits@
}

\newcommand\big@doubleop@aux[2]{%
  \sbox\z@{$\m@th#1#2$}%
  \makebox[1.35\wd\z@][s]{$\m@th#1#2\hss#2$}%
}

\newcommand{\abs}[1]{\left|#1\right|}     %usage: \abs{x} yields |x|.
\theoremstyle{plain}
\newtheorem{theorem}{Theorem}
\newtheorem{lemma}{Lemma}

\newtheorem{remark}{Remark}
\newtheorem{definition}{Definition}

\newtheorem{example}{Example}

\newtheorem{axiom}{Axiom}
\newtheorem{conjecture}{Conjecture}

\usepackage[title]{appendix}

\makeatletter
\@namedef{subjclassname@2020}{\textup{2020} Mathematics Subject Classification}
\makeatother

\begin{document}

\title{Energy Dissipation in Hilbert Envelopes on Motion Waveforms Detected in Vibrating Systems:\\
An Axiomatic Approach}

\author[J.F. Peters]{J.F. Peters}
\address{
Department of Electrical \& Computer Engineering,
University of Manitoba, WPG, MB, R3T 5V6, Canada and
Department of Mathematics, Faculty of Arts and Sciences, Ad\.{i}yaman University, 02040 Ad\.{i}yaman, T\"{u}kiye,
}
\email{james.peters3@umanitoba.ca}

\author[T.U. Liyanage]{T.U. Liyanage}
\address{
Department of Electrical \& Computer Engineering,
University of Manitoba, WPG, MB, R3T 5V6, Canada
}
\email{uswattat@myumanitoba.ca}

\subjclass[2020]{74H80 (Energy minimalization),76F20 (Dynamical systems), 93A05 (Axiomatic systems theory)}

\date{}

\begin{abstract}
This paper introduces an axiomatic approach in the theory of energy dissipation in Hilbert envelopes on waveforms emanating from various vibrating systems. A Hilbert envelope is a curve tangent to peak points on a motion waveform. The basic approach is to compare non-modulated vs. modulated waveforms in measuring energy loss during the vibratory motion $m(t)$ at time $t$ of moving object such as a walker, runner or biker recorded in a video. Modulation of $m(t)$ is achieved by using Mersenne primes to adjust the frequency $\omega$ in the Fourier transform $m(t)e^{\pm j2\pi \omega t}$.   Expediture of energy $E_{m(t)}$ by a system is measured in terms of the area bounded by the motion $m(t)$ waveform at time $t$.
\end{abstract}
%
%\keywords[2020 Mathematics Subject Classification]{054E05 (Proximity); 54C50 (Topology)}
\keywords{Dissipation, Energy, Frequency, Hilbert envelope, Mersenne prime, Motion waveform, Vibrating System, Video Frames}
\maketitle
\tableofcontents

\section{Introduction}
Dynamical system vibrations appear as varying oscillations in motion waveforms~\cite{DeLeoYork2024,Feldman2011vibration}. 
The focus in this paper is on the detection of energy dissipation that commonly occurs in vibrating dynamical systems.  For a motion waveform $m(t)$ at time $t$, measure of motion dissipated energy $E_{diss}$ is defined in terms of the difference between non-modulated energy $E_{nmod}$ and modulated energy $E_{mod}(t)$, i.e.,
\begin{align*}
E_{diss}(t) &= \abs{E_{nmod}(t)-E_{mod}(t)}\\
& = \abs{\mbox{non-modulated}\ E_{nmod}(t) - \mbox{modulated}\ E_{mod}(t)}. 
\end{align*} 
at time $t$ of a vibratory system.  In this work, two forms of motion waveform energy are considered, namely, non-modulated (smoothing) of $m(t)$ and modulated $m(t)$ that results from the product of $m(t)$ and the exponential $e^{\pm j2\pi \omega t}$ introduced by Euler~\cite{Euler1748}.  A formidable source of waveform energy measurement result from the Fourier transform $m(t)e^{\pm j2\pi \omega t}$~\cite{Fourier1822}. 

A non-moduled form of waveform energy $E_m(t)$ is associated with the planar area bounded by motion curve beginning at instant $t_0$ and ending instant $t_1$, namely, 
%\begin{center}
\boxed{\boldsymbol{
E_{m(t)} = \int_{t_0}^{t_1}\abs{m(t)}^2dt
}}. In other words, system energy is identified with system waveform area, instead of the more usual energy graph~\cite{Kok2017dissipation}.  Modulated system energy is measured using \boxed{\boldsymbol{
E_{m(t)} = \int_{t_0}^{t_1}\abs{m(t)e^{\pm j2\pi \omega t}}^2dt}}.
%\end{center}
%, instead of the more usual energy graph~\cite{Kok2017dissipation}.
%Vibrating system energy is measured using a Hilbert envelope on a vibration waveform
%\boxed{E_m = \int_{t_0}^{t_1}\abs{m(t)e^{\pm j2\pi \omega t}}^2dt}.  The source of vibrating system energy measurements is  instead of the more usual energy graph~\cite{Kok2017dissipation}.  
%One of the important findings in this paper is the advantageous use of Mersenne primes to adjust the frequency $\omega$ of the Euler exponential.

An application of the proposed approach in measuring energy dissipation is given in terms of the Hilbert envelope on the peak points on waveforms derived from the up-and-down movements of the up-and-down movements of a walker, runner or biker recorded in a sequence of video frames.  An important finding in this paper is the effective use of Mersenne primes to adjust the frequency $\omega$ of the Euler exponential to achieve waveform modulation with minimal energy dissipation (in the uniform waveform case (see Conjecture~\ref{conjecture}.\ref{case1}).  We prove that waveform energy is a characteristic, which maps to the complex plane (See Theorem~\ref{theorem:Em}. This result extends the waveform energy results in~\cite{Tiwari2024},~\cite{Peters2020}).

\begin{table}[h!]\label{table:symbols}
\begin{center}
\caption{Principal Symbols Used in this Paper}
\begin{tabular}{|c|c|}
%\toprule
\hline
Symbol & Meaning\\ 
\hline\hline
$2^A$ & Collection of subsets of a nonempty set $A$\\
\hline
$A_i\in 2^A$ & Subset $A_i$ that is a member of $2^A$\\
\hline
$\mathbb{C}$ & Complex plane\\
\hline
$t$ & Clock tick\\
\hline
$e^{j\omega t}$ & $cos(\omega t) + jsin(\omega t)$~\cite{Euler1748}\\
\hline
$M$ & Mersenne prime\\
\hline
$\omega$ & Waveform Oscillation Frequency\\
\hline
$E_{m(t)}$ & Energy of motion waveform $m(t)$\\
\hline
$\varphi_t:2^A\to \mathbb{C}$ & $\varphi$ maps $2^A$ to complex plane $\mathbb{C}$ at time $t$\\
\hline
$\varphi_t(A_i\in 2^A)\in\mathbb{C}$ & Characteristic of $\varphi(A_i\in 2^A)\in\mathbb{C}$ at time $t$.\\
\hline
\end{tabular}
\end{center}
\end{table}

\section{Preliminaries}
Highly oscillatory, non-periodic waveforms provide a portrait of vibrating system behavior. Energy dissipation (decay) is a common characteristic of every vibrating dynamical system. Included in this paper is an axiomatic basis is given for measuring this characteristic of dynamical systems.  A characteristic is a mapping $\varphi_t:A_i\to \mathbb{C}$, which maps a subsystem $A_i$ in a system $A$ to a point in the complex plane $\mathbb{C}$.

\begin{figure}[!ht]
	\centering
	\includegraphics[width=65mm]{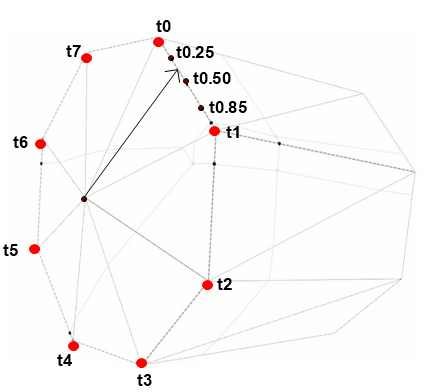}
	\caption{Morse instants clock}
	\label{fig:MorseClock}
\end{figure}  

\begin{definition}\label{def:system0}{\bf(System)}.\\
A {\bf system} $A$ is a collection of interconnected components (subsystems $A_i\in 2^A$) with input-output relationships.
\qquad \textcolor{blue}{$\blacksquare$}
\end{definition}
%\vspace*{0.2cm}

\begin{definition}\label{def:dynamicalSystem0}{\bf(Dynamical System)}.\\
A {\bf dynamical system} is a time-constrained physical system.
\qquad \textcolor{blue}{$\blacksquare$}
\end{definition}

\begin{definition}\label{def:dSwaveform}{\bf(Dynamical System Output Waveform)}.\\
The output of a {\bf dynamical system} is a time-constrained sequence of discrete values.
\qquad \textcolor{blue}{$\blacksquare$}
\end{definition}

It has been observed that the theory of dynamical systems is a major mathematical discipline closely aligned with many areas of mathematics ~\cite{Katok1995dynSys}.  Energy dissipation is considered in many contexts such as heating, liquid (viscosity) and water-wave scattering. In this work, the focus is on energy decay represented by the difference between the energy of  non-smooth (non-modulated) and smooth (modulated) motion waveforms.  A motion waveform is a graphical portrait of the radiation emitted by moving system (e.g., walker, runner, biker) with oscillatory output.

\begin{axiom}\label{axiom:clock}{\bf (Instants Clock)}.\\
Every system has its own instants clock, which is a cyclic mechanism that is a simple closed curve with an instant hand with one end of the instant hand at the centroid of the cycle and the other end tangent to a curve point indicating an elapsed time in the motion of a vibrating system.  A clock tick occurs at every instant that a system changes its state. 
\qquad \textcolor{blue}{$\blacksquare$}
\end{axiom}
\vspace*{0.2cm}

\begin{remark}{\bf (What Euler tells us about time)}.\\
On an instants clock, every reading $t\in \mathbb(C)$, a point \boxed{t = a+jb, a,b\in\mathbb{R}} in the complex plane.
For example, \boxed{t_{0.25}=0.25 + j0 = 0.25} in Fig.~\ref{fig:MorseClock}. The Morse instants clock is also called a homographic clock~\cite{HofmannKasner1928}, since the tip of an instant clock $t$-hand moves on the circumference a circle, where $t$ is a complex number~\cite{Kasner1928}. For $t$ at the tip of a vector with radius $r=1$, angle $\theta$ and \boxed{a = cos\theta,b=sin\theta} in the complex plane, then
\begin{center}
\boxed{\boldsymbol{
t = a+jb = cos\theta + jsin\theta = e^{j\theta}.
}}
\end{center}
\vspace*{0.2cm}
An instant of time viewed as an exponential is inspired by Euler~\cite{Euler1748}. \qquad \textcolor{blue}{$\blacksquare$}
\end{remark}
\vspace*{0.2cm}

\begin{example}
A sample Morse instants clock is shown in Fig.~\ref{fig:MorseClock}. The clock hand points to the elapsed time in the interval \boxed{\boldsymbol{(t_{0.25}\leq t\leq t_{0.25})}} in milliseconds (ms) after a system has begun vibrating. The clock face is a polyhedral surface in a Morse-Smale in a convex polyhedron in 3D Euclidean space~\cite{Ludmany2023}, since Morse-Smale polyhedron is an example of a mechanical shape descriptor ideally suited as clock model because of its underlying piecewise smooth geometry. This form of an instants clock has been chosen to emphasize that the elapsed time $t_k$ is a real number in an instants interval \boxed{\left[t_0,t_k\right]\in\mathbb{C}^2} in which $t_k$ is indeterminate. From a planar perspective, the proximity of sets of instants clock times is related to results given for computational proximity in the digital plane~\cite{Peters2019}.  In this example, the instant hand is pointing to an elapsed time between $t_{0.25}$ ms and $t_{0.50}$ ms.
\qquad \textcolor{blue}{$\blacksquare$}
\end{example}
\vspace*{0.2cm}

\begin{definition}\label{def:characteristic}{\bf (Clocked Characteristic of a subsystem)}.\\
The clocked characteristic of a subsystem $A_i$ of a system $A$ at time $\varphi_t(A_i)$ is a mapping \boxed{\varphi_t:A_i\in 2^A\to \mathbb{C}} defined by
$\varphi_t(A_i)=a+bj\in\mathbb{C}, a,b\in\mathbb{R},\\ j = \sqrt(-1),\varphi_t(A_i)\in\mathbb{C}$.
\qquad \textcolor{blue}{$\blacksquare$}
\end{definition}
\vspace*{0.2cm}

\begin{mdframed}[backgroundcolor=green!15]
\label{axiom:+ve char}{\bf  Vibrating System Time-Constrained Motion Characteristic.}

\begin{axiom}\label{axiom:complexNo}{\bf (Subsystem Motion Characteristic)}.\\
Let $A_i\in 2^A$ (subsystem $A_i$ in the collection of subsystems $2^A$ in system $A$) that emits changing radiation due to system movements (motion) and let $t$ be a clock tick.  The motion characteristic of subsystem motion $A_i\in 2^A$ is a mapping $m_t:A_i\to \mathbb{C}$ defined by
%\vspace*{0.2cm}
\begin{center}
\boxed{\boldsymbol{ 
m_t(A_i)=a+bj\in \mathbb{C},a,b\in \mathbb{R}, j = \sqrt(-1),t\in\mathbb{R}.
}}
\vspace*{0.2cm}
\end{center} 
%\vspace*{0.2cm}
i.e., a subsystem $A_i$ motion characteristic of a system $A$ is a mapping $m_t(A_i\in 2^A)\in\mathbb{C}$ at time $t$. 
\qquad \textcolor{blue}{$\blacksquare$}
\end{axiom}
\vspace*{0.2cm}

\begin{remark}
For the motion characteristic, we write \boxed{m(t)} when it is understood that motion is on a subset $A_i\in 2^A$ in a dynamical system $A$.  Axiom~\ref{axiom:complexNo} is consistent with the view~\cite[p. 81]{Blair1976} of the characteristic vector field, represented here with a planer characteristic vector field $\boldsymbol{\xi}$ of a dynamical system with points $p(x,y,t)\in \xi$ that has positive complex characteristic coordinates at clock tick (time) $t$ such that\\
\begin{center}
\boxed{
\varphi_t(A_i\in 2^A) = 
p\in \xi = 
(a-jb)\frac{\partial \xi}{\partial x}
+ (a-jb)\frac{\partial \xi}{\partial y}
+ (a-jb)\frac{\partial \xi}{\partial t}, a,b\in\mathbb{R}.
}
\vspace*{0.2cm}
\end{center}
The 1-1 correspondence between every point $p$ having coordinates in the Euclidean plane and points in the complex plane is lucidly introduced by D. Hilbert and S. Cohn-Vossen~\cite[\S 38, 263-265]{Hilbert1952}.  For an introduction to characteristic groups, see~\cite{Tziolas2024},\cite{Bombieri1976},\cite{Lang1979}.  
\qquad \textcolor{blue}{$\blacksquare$}
\end{remark}
\end{mdframed}
%\vspace*{0.2cm}

\begin{figure}[!ht]
	\centering
	\includegraphics[width=85mm]{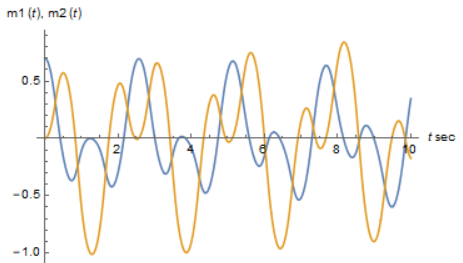}
	\caption{Sample spring system waveforms}
	\label{fig:springPortrait}
\end{figure} 

\begin{example}{\bf Spring system vibration}.\\
A pair of sample sinusoidal waveforms emitted by an expanding and contracting spring system is shown in Fig.~\ref{fig:springPortrait}.
\qquad \textcolor{blue}{$\blacksquare$}
\end{example}
\vspace*{0.2cm}

Vibrating system waveform $m(t)$ modulation (smoothing) is achieved by adjusting the frequency $\omega$ in an Euler exponential \boxed{e^{\pm j2\pi \omega t}}, which is used in oscillatory waveform curve smoothing. It has been found that 
Mersenne primes provide an effective means an effective means of adjusting the frequency $\omega$.   It has been observed by G.W. Hill~\cite{Hill1979} that Mersenne primes $M_p = 2^p - 1 = 3,7,31,...$ for Prime $p=2,3,5,...$ are useful in estimating variability as well as in estimating average values in sequences of discrete values.
\vspace*{0.2cm}

\begin{axiom}\label{axiom:waveformEnergy}{\bf (Waveform Energy)}.\\
A measure of dynamical system energy is the area of a finite planar region bounded by system waveform $m(t)$ curve at time $t$,
defined by
%\vspace*{0.2cm}
\begin{center}
\vspace*{0.2cm}
\boxed{\boldsymbol{
E_{m(t)} = \int_{t_0}^{t_1}\abs{m(t)}^2dt
}.}
\end{center}
\end{axiom}
\vspace*{0.2cm}

\begin{lemma}\label{lemma:energy}
Dynamical system energy is time-constrained and is always limited.
\end{lemma}
\begin{proof}
Let $E_m$ be the energy of a dynamical system, defined in Axiom~\ref{axiom:waveformEnergy}. From Axiom~\ref{axiom:waveformEnergy}, system energy always occurs in a bounded temporal interval $\left[t_0,t_1\right]$.
Hence, $E_m$ is time constrained.
From Axiom~\ref{axiom:clock}, the length of a system waveform is finite, since, from Axiom~\ref{axiom:waveformEnergy}, system duration is finite. From Axiom~\ref{axiom:waveformEnergy}, system energy is derived from the area of a finite, bounded region.  Consequently, system energy is always finite.
\end{proof}
\vspace*{0.2cm}

\begin{theorem}\label{theorem:Em}
If $X$ is a dynamical system with waveform $m(t)$ at time $t$ and which changes with every clock tick, then observe
\begin{compactenum}[1$^o$]
\item System waveform characteristic values are in the complex plane.
\item\label{step:part1} System energy varies with every clock tick.
\item System radiation characteristics are finite.
\item All system characteristics map to the complex plane.
\item Waveform energy decay is a characteristic, which maps to $\mathbb{C}$.
\end{compactenum}
\end{theorem}
\begin{proof}$\mbox{}$\\
\begin{compactenum}[1$^o$]
\item From Def.~\ref{def:characteristic}, a system characteristic is a mapping from a subsystem to the complex plane at time $t$,  From Axiom~\ref{axiom:complexNo}, every waveform motion characteristic $m(t)\in \mathbb{C}$ at time $t$, which is the desired result.
\item From Lemma~\ref{lemma:energy}, system energy is time-constrained and always occurs in a bounded temporal interval.  From Axiom~\ref{axiom:clock}, there is a new clock tick at every instant in time $t$ ms. From Axiom~\ref{axiom:waveformEnergy}, system energy varies with every clock tick.
\item From Axiom~\ref{axiom:clock}, all system radiation characteristics are finite, since system duration is finite.
\item\label{step:characteristic} From Axiom~\ref{axiom:complexNo}, every system $A$ characteristic is a mapping from a subsystem $A_i\in 2^A$ to the complex plane, which is the desired result.
\item From the proof of step~\ref{step:characteristic}, the desired result follows.
\end{compactenum}
\end{proof}
\vspace*{0.2cm}

\begin{figure}[!ht]
\centering
\includegraphics[width=80mm]{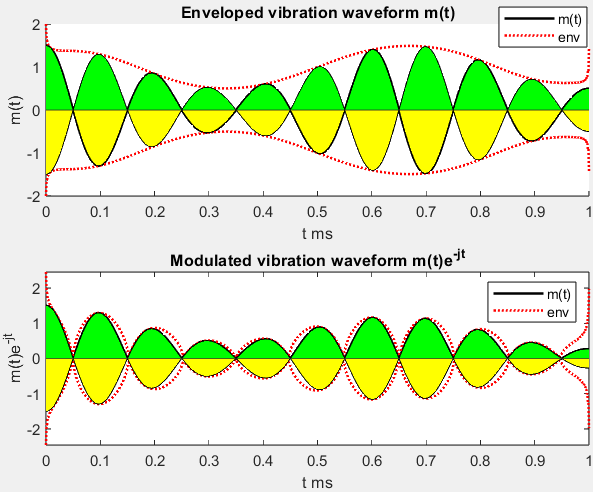}
\caption{Hilbert envelope on modulated vibration waveform.}
\label{fig:lobes}
\end{figure}

To obtain an approximation of system energy, a system waveform is represented by a continuous curve defined by a Hilbert envelope tangent to waveform peak points, forming what known as Hilbert lobes.
\vspace*{0.2cm}

\begin{mdframed}[backgroundcolor=green!15]{\bf Hilbert envelope lobes.}\label{def:HilbertLobes}\\
A {\bf Hilbert envelope} (denoted by \boxed{H_{env}}) is a curve that is tangent to the peak points on a waveform~\cite[\S 18.4, p. 132]{King2009Hilbert}.  A {\bf Hilbert envelope lobe} (denoted by \boxed{H_{env_{lobe}}}) is a tiny bounded planar region attached to single waveform peak point on a waveform envelope, defined by
\[
\boldsymbol{H_{env} = \sqrt{m(t)^2 + (-m(t))^2}}\mbox{~\cite{Brandt2011}}
\]
The energy represented by a lobe $\boldsymbol{H_{env_{lobe}}}$ area of a tiny planar region attached to an oscillatory motion waveform $m(t)$ is defined by 
\[
H_{env_{lobe}} = \int_{a}^{b}\abs{m(t)}^2dt
\]
It is lobe area that provides a measure of the energy repesented by a waveform segment.
\end{mdframed}
\vspace*{0.2cm}
%\begin{figure}[!ht]
%\centering
%\includegraphics[width=90mm]{absEnv2}
%\caption{Hilbert envelope on modulated vibration waveform.}
%\label{fig:lobes}
%\end{figure}

The modulated vibration waveform $m(t)$ in Fig.~\ref{fig:lobes} varies with lower peak points than the original motion waveform, depending on the choice of Mersenne prime frequency.  To minimize energy loss due to modulation, a Mersenne prime is chosen for the frequency $\omega$ in an Euler exponential in \boxed{\boldsymbol{m(t)e^{\pm j\omega t}}} to obtain
\begin{compactenum}[{\bf result}.1$^o$]
\item Modulated system waveform $m(t)$ is smoother for a particular Mersenne prime frequency (i.e., the waveform oscillations are more uniform).
\item Modulated system energy loss is minimal,  for a particular Mersenne prime frequency.
\end{compactenum}

\section{Application: Modulating System Waveform with Minimal Energy Dissipation}
In this section, we illustrate how Mersenne primes can be used effectively to obtain the following results:
\begin{compactenum}[{M$\to \omega$-}1$^o$]
\item Usage of a M-prime as the frequency in the Euler exponential in\\
 
\begin{center}
\vspace*{0.2cm}
\boxed{\boldsymbol{m(t)e^{\pm j\omega Mt}}}
\end{center}
\vspace*{0.2cm}

reduces motion $m(t)$ waveform motion energy. 
\item Energy dissipation varies in modulated vs. non-modulated waveforms for different choices of frequency $M$ in \boxed{\boldsymbol{e^{\pm j\omega Mt}}}, depending on whether a waveform has uniformly or non-uniformly varying cycles around the origin.
\end{compactenum}

\begin{figure}[!ht]
\centering
\includegraphics[width=80mm]{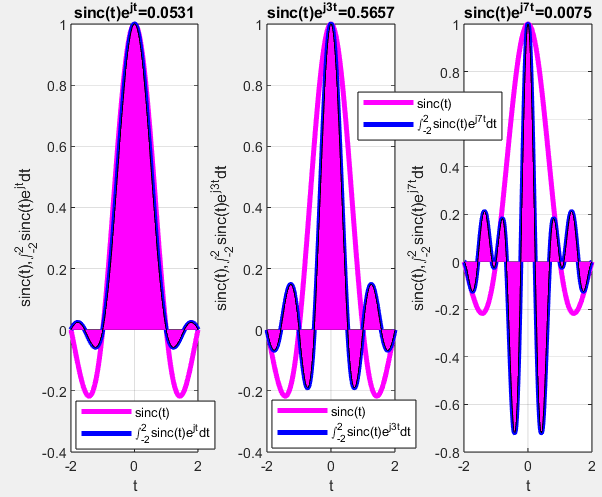}
\caption{3 forms of $m(t)e^{j\omega t}$}
\label{fig:3forms}
\end{figure}

\begin{conjecture}\label{conjecture}
The choice of a Mersenne prime \boxed{M\leq 31} will always result in lower motion waveform peak values using $M$ as the frequency in the Euler exponential to achieve waveform modulation and minimal energy dissipation.
\end{conjecture}
There are two cases to consider:
[{\bf Partial Picture Proof}]$\mbox{}$\\
\begin{compactenum}[{Case}(i)]
\item\label{case1}  
Assume $m(t)$ waveform uniformly fluctuates and frequency $\omega = M = 1$ results in the lowest energy loss.\\
\begin{proof}
Partial picture proof:
Recall that $e^{j\omega t} = cos\omega t+jsin\omega t$,
where $m(t)e^{j\omega t}$ forces the oscillation in a motion waveform to increase.
Let \boxed{m(t)=sinc(t)}, introduced in 1822 by Fourier~\cite{Fourier1822}.  Then $m(t)$ oscillates uniformly on either side of the origin (see sample plot of $sinc t$ in Fig.~\ref{fig:3forms}). The area of \boxed{m(t)=\int_{-k}^{k}sinc(t)e^{j\omega t}dt,\omega \geq 1} is always less than the area \boxed{\int_{-k}^{k}sinc(t)}, $e^{j\omega t}$ partitions $m(t)$ each cycle into regions with smaller areas whose total areas is less than the total area $\int_{-k}^{k}sinc(t)dt$. With $\omega = M = 1$, the modulated waveform energy is closest to non-modulated waveform energy, which is the desired result.
\end{proof}
\item Let $m(t)$ be a non-uniform waveform. We make the unproved claim that the choice of $\omega = M$, varies, i.e., $M$ is not always 1.
\end{compactenum}

\begin{example}{\bf Sample Energy Dissipation: Non-uniform waveform Case~\cite{Tharaka2024thesis}}.\\
Let $m(t) = sinc t$, with cycles that vary uniformly relative to the origin.  This is the case in Fig.~\ref{fig:3forms}.  The result for 3 choices of $M\in\left\{1,3,7\right\}$ are shown in the plots in Fig.~\ref{fig:3forms}.
This leads to the following energy dissipation levels:
\begin{align*}
E_{m(t)} &= 0.9028\ \mbox{non-modulated waveform energy}\\
E_{m(t)e^{jt}} &= 0.1503\ \mbox{$M=1$,modulated waveform energy loss}\\
E_{m(t)e^{j3t}} &= 0.3371\ \mbox{$M=3$,modulated waveform energy loss}\\
E_{m(t)e^{j7t}} &= 0.8954\ \mbox{$M=7$,modulated waveform energy loss}\\
E_{m(t)e^{j31t}} &= 0.9021\ \mbox{$M=31$,modulated waveform energy loss}
\end{align*}
The $\omega=M=31$ case is not shown in Fig.~\ref{fig:3forms}.
\qquad \textcolor{blue}{$\blacksquare$}
\end{example}
%\vspace*{0.2cm}

Evidence of the correctness of our Conjecture for the non-uniform waveform case in the choice of the Mersenne prime to achieve minimal energy dissipation can be seen in the following two examples.

\begin{figure}[ht!] % 'H' ensures the figure is placed exactly here
    \centering
    \includegraphics[width=115mm]{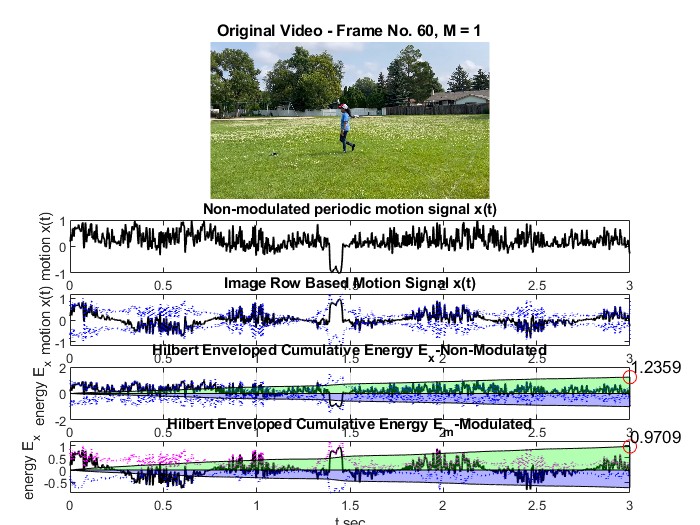} % Use width to ensure full image display
    \caption{Energy Analysis for M = 1 of Walker: Comparison of Non-Modulated and Modulated Motion Signals of Frame = 60}
		\label{fig:walker}
\end{figure}

\begin{example}\label{ex:walker}{\bf Sample Energy Dissipation for a walker waveform}.\\
A sample collection of non-modulated and modulated waveforms for a walker for M = 1 is shown in Fig.~\ref{fig:walker}.  In Table~2, M = 1
for the exponential frequency of a modulated waveform
results in the lowest energy dissipation.
\qquad \textcolor{blue}{$\blacksquare$}
\end{example}
%\vspace*{0.2cm}

However, if consider the choice of $M$ for the modulation frequency for a biker, this choice differs from the choice of $M=1$ in Example~\ref{ex:walker}.

\begin{table}[htbp]\label{table:walker}
\centering
\caption{Energy Dissipation for Walking}
\begin{tabular}{|c|c|c|c|}
\hline
\textbf{M} & \textbf{Non-Modulated Energy} & \textbf{Modulated Energy} & \textbf{Energy Dissipation} \\
           & \textbf{(Ex)}                 & \textbf{(Em)}             & \textbf{Percentage}         \\
\hline
1  & 1.2359 & 0.9709 & 21.44\% \\
\hline
3  & 1.2359 & 0.9222 & 25.38\% \\
\hline
7  & 1.2359 & 0.9559 & 22.65\% \\
\hline
31 & 1.2359 & 0.9166 & 25.83\% \\
\hline
\end{tabular}
\label{tab:walk_energy_dissipation}
\end{table}

\begin{figure}[ht!] % 'H' ensures the figure is placed exactly here
    \centering
    \includegraphics[width=115mm]{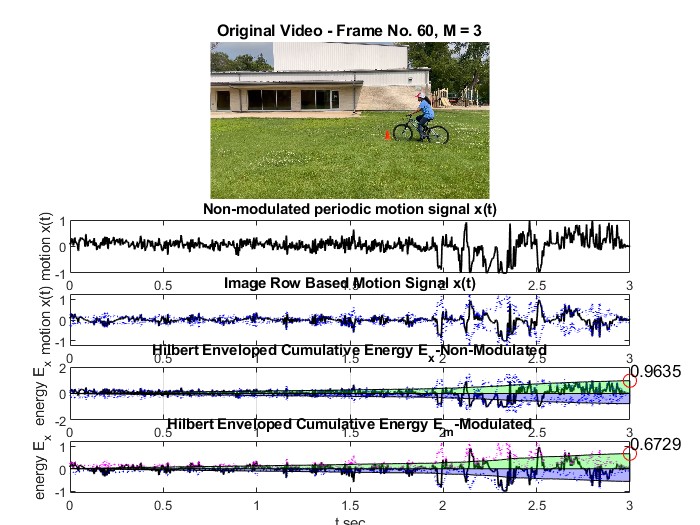} % Use width to ensure full image display
    \caption{Energy Analysis for M = 3 of Biker: Comparison of Non-Modulated and Modulated Motion Signals of Frame = 60}
		\label{fig:biker}
\end{figure}

\begin{table}[htbp]\label{table:energy2}
\centering
\caption{Energy Dissipation for Biking}
\begin{tabular}{|c|c|c|c|}
\hline
\textbf{M} & \textbf{Non-Modulated Energy} & \textbf{Modulated Energy} & \textbf{Energy Dissipation} \\
           & \textbf{(Ex)}                 & \textbf{(Em)}             & \textbf{Percentage}         \\
\hline
1  & 0.9635 & 0.6317 & 34.43\% \\
\hline
3  & 0.9635 & 0.6729 & 30.16\% \\
\hline
7  & 0.9635 & 0.6561 & 31.90\% \\
\hline
31 & 0.9635 & 0.6396 & 33.61\% \\
\hline
\end{tabular}
\label{tab:bike_energy_dissipation}
\end{table}

\begin{example}\label{ex:biker}{\bf Sample Energy Dissipation for a biker waveform}.\\
A sample collection of non-modulated and modulated waveforms for a biker for M = 3 is shown in Fig.~\ref{fig:biker}.  In Table~3, M = 3
for the exponential frequency of a modulated waveform
results in the lowest energy dissipation.
\qquad \textcolor{blue}{$\blacksquare$}
\end{example}
%\vspace*{0.2cm}

%\begin{table}[htbp]\label{table:energy2}
%\centering
%\caption{Energy Dissipation for Biking}
%\begin{tabular}{|c|c|c|c|}
%\hline
%\textbf{M} & \textbf{Non-Modulated Energy} & \textbf{Modulated Energy} & \textbf{Energy Dissipation} \\
           %& \textbf{(Ex)}                 & \textbf{(Em)}             & \textbf{Percentage}         \\
%\hline
%1  & 0.9635 & 0.6317 & 34.43\% \\
%\hline
%3  & 0.9635 & 0.6729 & 30.16\% \\
%\hline
%7  & 0.9635 & 0.6561 & 31.90\% \\
%\hline
%31 & 0.9635 & 0.6396 & 33.61\% \\
%\hline
%\end{tabular}
%\label{tab:bike_energy_dissipation}
%\end{table}

\section{Conclusion}
This paper focuses on the frequency characteristic in  modulating dynamical system waveforms.  The appropriate choice of Mersenne prime $M$ as the frequency $\omega$ for  the Euler expontial $e^{j\omega t}\to e^{jMt}$ is considered in modulating a dynamical system waveform to obtain a smoother waveform and achieve minimal energy dissipation, applying Mersenne primes.  It has been found that $M=1$ is the best choice for waveforms whose cycles vary uniformly about the origin.  Choice of $M\in\left\{1.3,7,31\right\}$ for the non-uniform waveforms varies, depending on how extreme the lack of self-similarity is present in waveforms that vary in a chaotic fashion on either side of the origin.  The appropropriate choice of $M$ in modulating a non-uniform waveform is an open problem.
\vspace*{0.2cm}

\section*{Acknowledgements}
This research has been supported by the Natural Sciences \&
Engineering Research Council of Canada (NSERC) discovery grant 185986 
and Instituto Nazionale di Alta Matematica (INdAM) Francesco Severi, Gruppo Nazionale per le Strutture Algebriche, Geometriche e Loro Applicazioni grant 9 920160 000362, n.prot U 2016/000036 and Scientific and Technological Research Council of Turkey (T\"{U}B\.{I}TAK) Scientific Human
Resources Development (BIDEB) under grant no: 2221-1059B211301223.
\vspace{0.2cm}

%\vspace{0.2cm}

\bibliographystyle{amsplain}
\bibliography{NSrefs}

\end{document}